\documentclass[10pt, conference]{IEEEtran}

\usepackage[pdftex,colorlinks=true,linkcolor=black,citecolor=black,urlcolor=blue]{hyperref}

\usepackage{amsmath,amssymb,amsthm}
\usepackage{booktabs}
\usepackage{hyperref}

\usepackage{subfig}

\usepackage{textcomp} 
\usepackage{stmaryrd} 

\usepackage[capitalize,noabbrev,nameinlink]{cleveref} 

\usepackage{aliascnt}

\crefformat{equation}{#2Equation~#1#3}

\usepackage{thmtools,thm-restate}

\newtheorem{theorem}{Theorem}

\newtheorem{lemma}[theorem]{Lemma}

\theoremstyle{definition} 
\newtheorem{definition}{Definition}

\usepackage{url}

\theoremstyle{remark} 

\usepackage{mathpartir}

\usepackage{booktabs}

\usepackage{xcolor}
\definecolor{ltblue}{rgb}{0,0.4,0.4}
\definecolor{dkblue}{rgb}{0,0.1,0.6}
\definecolor{dkgreen}{rgb}{0,0.35,0}
\definecolor{dkviolet}{rgb}{0.3,0,0.5}
\definecolor{dkred}{rgb}{0.5,0,0}

\usepackage{listings}

\usepackage[noend]{algpseudocode}

\usepackage{tikz} 
\usetikzlibrary{cd}
\usetikzlibrary{positioning}
\usetikzlibrary{shapes.multipart}
\usetikzlibrary{%
  arrows,%
  shapes.misc,
  shapes.arrows,%
  shapes.callouts,
  chains,%
  matrix,%
  positioning,
  scopes,%
  decorations.pathmorphing,
  decorations.text,
  shadows%
}
\usetikzlibrary{quantikz}

\usepackage{textcomp} 

\usepackage{xparse}
\NewDocumentCommand{\optionalParens}{s m m}{
    \IfBooleanTF{#2}{\left(#3\right)}{\IfBooleanTF{#1}{~#3}{#3}}
}
\NewDocumentCommand{\apply}{m s O{} m}{
    #1#3 \optionalParens*{#2}{#4}
}

\NewDocumentCommand{\applytwo}{m O{} s m s m}{ 
   #1#2~\optionalParens{#3}{#4}~\optionalParens{#5}{#6}
}

\usepackage{xspace}

\usepackage{braket}
\usepackage{blochsphere}

\usepackage{comment}
\usepackage{style/draft}

\drafttrue

\newnote[JP]{jennifer}{red}
\newnote[AS]{albert}{blue}
\newnote[AM]{anne}{purple}
\newnote[GG]{gian}{purple}
\newnote{fixme}{green}
\newnote[Q]{question}{purple}
\newnote{note}{green}
\newnote{todo}{green}
\newnote[Cite:]{tocite}{green}

\newcommand{\ShortVersion}
    { \includecomment{ShortVersion}
      \excludecomment{ExtendedVersion}
    }
    
\usepackage{enumerate}
\usepackage[shortlabels,inline]{enumitem}

\newcommand{\companion}{\cite{PaykinSchmitz2023PCOAST}\xspace}

\newcommand{\PCAST}{\text{PCOAST}\xspace}

\usepackage{xcolor}

\newcommand{\poprgate}{\gate[style={rounded corners}]}
\newcommand{\startingpoprgate}{\gate[nwires={1},style={rounded corners}]}

\newcommand{\startingframegate}[2]{\gate[nwires={1,2}, wires=#1, style={rounded corners}]{#2} }

\newcommand{\PrepZ}{\texttt{PrepZ}}
\newcommand{\MeasZ}{\texttt{MeasZ}}
\newcommand{\MeasX}{\texttt{MeasX}}
\newcommand{\RX}[1]{\texttt{RX}(#1)}

\newcommand{\support}{\textsf{supp}}

\NewDocumentCommand\interp{sm}
        {\IfBooleanTF{#1}
                {\conjugate{\llbracket #2 \rrbracket}}
                {\llbracket #2 \rrbracket}
        }

\usepackage{MnSymbol} 
\NewDocumentCommand\commute{smm}
        {#2 \IfBooleanT{#1}{\not}\upmodels #3}

\NewDocumentCommand{\commutativity}{mm}{\lambda(#1,#2)}

\newcommand{\mc}{\mathcal}
\newcommand{\mb}{\mathbb}
\newcommand{\Span}{\text{Span}}
\newcommand{\mask}{\textsf{Mask}}
\newcommand{\stabs}{\textsf{Stabs}}
\renewcommand{\proj}{\textsf{Proj}}
\newcommand{\tr}{\textsf{Tr}}
\newcommand{\decouple}{\textsf{Decouple}}
\newcommand{\cost}{\textsf{Cost}}

\newcommand{\CNOT}{\textsf{CNOT}}

\newcommand{\X}{\text{X}}
\newcommand{\eff}[1]{\textsf{eff}#1}
\newcommand{\effX}{\textsf{eff}X}
\newcommand{\effZ}{\textsf{eff}Z}
\newcommand{\effY}{\textsf{eff}Y}

\newcommand{\pprep}{\textsf{Prep}\xspace}
\newcommand{\mmeas}{\textsf{Meas}\xspace}
\newcommand{\Q}{\text{Q}}
\newcommand{\C}{\text{C}}

\newcommand{\reduceNode}{\textsc{reduceNode}\xspace}

\newcommand{\fwdAction}[1]{\overrightarrow{#1}}
\newcommand{\bwdAction}[1]{\overleftarrow{#1}}

\NewDocumentCommand{\Rot}{t2 m m o o}
    {\IfBooleanTF{#1}
        {\textsf{Rot}^2(#2, #3, #4, #5)}
        {\textsf{Rot}(#2,#3)}}
\newcommand{\Prep}[2]{\textsf{Prep}(#1,#2)}
\NewDocumentCommand{\Meas}{O{} m}{
        \textsf{Meas}^{#1}(#2)
}
\NewDocumentCommand{\conjugate}{sm}{{\optionalParens{#1}{#2}}^\ast}

\NewDocumentCommand{\pauliToFrame}{t2 m}
        {\IfBooleanTF{#1}
                {F^{#2, \tfrac{\pi}{2}}}
                {F^{#2}}
        }

\NewDocumentCommand\cosBy{t2 m}
        {\IfBooleanTF{#1}
                {\cos{(\tfrac{#2}2)}}
                {\cos{#2}}
        }
\NewDocumentCommand\sinBy{t2 m}
        {\IfBooleanTF{#1}
                {\sin{(\tfrac{#2}2)}}
                {\sin{#2}}
        }

\NewDocumentCommand{\commuteTTRule}{o}
        {\IfNoValueTF{#1}
                {\ensuremath{\bot^t}}
                {\textsc{#1-\ensuremath{\bot^t}}}}
\NewDocumentCommand{\commuteTPRule}{o}
        {\IfNoValueTF{#1}
                {\ensuremath{\bot^p}}
                {\textsc{#1-\ensuremath{\bot^p}}}}

\NewDocumentCommand{\mergestep}{s}
        {\longrightarrow_M \IfBooleanT{#1}{^\ast}}
\NewDocumentCommand{\framestep}{s}
        {\longrightarrow_F \IfBooleanT{#1}{^\ast}}

\DeclareMathOperator*{\argmin}{arg\,min} 

\title{PCOAST:
        (Extended Version) 
}

\title{Optimization at the Interface of Unitary and Non-unitary Quantum Operations in \PCAST
}

\author{
    \IEEEauthorblockN{
        Albert T. Schmitz\IEEEauthorrefmark{1}\IEEEauthorrefmark{3},
        Mohannad Ibrahim\IEEEauthorrefmark{1},
        Nicolas P. D. Sawaya\IEEEauthorrefmark{2},
        Gian Giacomo Guerreschi\IEEEauthorrefmark{2}, \\
        Jennifer Paykin\IEEEauthorrefmark{1},
        Xin-Chuan Wu\IEEEauthorrefmark{2},
        and
        A. Y. Matsuura\IEEEauthorrefmark{1}
        }
    \IEEEauthorblockA{\IEEEauthorrefmark{1}
    \textit{Intel Labs, Intel Corporation}, 
    Hillsboro, OR, USA \\
    }
    \IEEEauthorblockA{\IEEEauthorrefmark{2}
    \textit{Intel Labs, Intel Corporation},
    Santa Clara, CA, USA\\
    }
    
    \IEEEauthorblockA{\IEEEauthorrefmark{3}
    Email: albert.schmitz@intel.com }
}

\usepackage[numbers]{natbib}
\begin{document}

\includecomment{DraftVersion} 
\ShortVersion

\maketitle

\begin{DraftVersion}
\thispagestyle{plain}
\pagestyle{plain}
\end{DraftVersion}

\begin{abstract}
The Pauli-based Circuit Optimization, Analysis and Synthesis Toolchain (PCOAST) was recently introduced as a framework for optimizing quantum circuits. It converts a quantum circuit to a Pauli-based graph representation and provides a set of optimization subroutines to manipulate that internal representation as well as methods for re-synthesizing back to a quantum circuit. In this paper, we focus on the set of subroutines which look to optimize the PCOAST graph in cases involving unitary and non-unitary operations as represented by nodes in the graph. This includes reduction of node cost and node number in the presence of preparation nodes, reduction of cost for Clifford operations in the presence of preparations, and measurement cost reduction using Clifford operations and the classical remapping of measurement outcomes. These routines can also be combined to amplify their effectiveness.

We evaluate the PCOAST optimization subroutines using the Intel$^{\text{\textregistered}}$\xspace Quantum SDK on examples of the Variational Quantum Eigensolver (VQE) algorithm. This includes synthesizing a circuit for the simultaneous measurement of a mutually commuting set of Pauli operators. We find for such measurement circuits the overall average ratio of the maximum theoretical number of two-qubit gates to the actual number of two-qubit gates used by our method to be 7.91.
\end{abstract}


\section{Introduction}

\PCAST, a Pauli-based Circuit Optimization, Analysis, and Synthesis Toolchain, was recently introduced as a comprehensive framework for quantum circuit optimization~\companion. It draws on a class of unitary quantum circuit optimizations based on Pauli rotations~\citep{Zhang2019,cowtan2019phase,Schmitz2021} that take advantage of the fact that unitary circuits can be decomposed into Clifford gates  and Pauli rotations $\Rot{P}{\theta}=e^{-i \theta/2 P}$, where $P$ is a Pauli matrix $X$, $Y$, or $Z$.  Because Cliffords always map Paulis to Paulis by conjugation, it is possible to push Clifford unitaries $U$ past the non-Clifford Pauli rotations $\Rot{P}{\theta}$ to produce a new rotation $\Rot{U P U^\dagger}{\theta}$. Doing so can expose the fact that some rotations can be merged~\citep{Zhang2019}.

\PCAST extends these optimizations into a comprehensive framework in several key ways. First, it introduces the \PCAST graph whose nodes represent quantum operations with edges based on the non-commutativity of the underlying Pauli operators. While prior work was limited to unitary optimizations, \PCAST acts on mixed unitary and non-unitary circuits by introducing preparation and measurement nodes parameterized by Paulis alongside unitary rotations and Clifford nodes. Second, \PCAST introduces a customizeable greedy synthesis algorithm to synthesize efficient gate representations from the \PCAST graph. And finally, \PCAST implements sophisticated internal optimizations on the \PCAST graph, which are the focus of this paper.

This paper presents the details of those \PCAST-to-\PCAST graph optimizations primarily aimed at exploiting the interaction between unitary and non-unitary nodes. These optimizations include reducing both the cost and number of non-Clifford nodes in the presence of preparations and the cost of Clifford operations in the presence of preparations. However, the most impactful of these optimizations is the reduction of mutually commuting measurements and their cost using Clifford operations and classical measurement remapping. Though other methods exist \cite{Verteletskyi2020,Gokhale2020, Crawford2021}, we demonstrate that our methods:  
\begin{itemize}
\item make no additional restriction on the measurement set;
\item leverage full and partial qubit agreement (see Defn.~\ref{def:agree});
\item have no limit on size or independence of the Pauli measurement set;
\item leverage lower-weight elements in the span of the measurement set to minimize the circuit costs; and
\item are fully integrated into the \PCAST framework, and inherit its efficiencies, refinements, and scalability\cite{PaykinSchmitz2023PCOAST} including future development. 
\end{itemize}
All methods discussed throughout the paper are implemented in the Intel$^{\text{\textregistered}}\xspace$ Quantum Software Development Kit (SDK)\footnote{\url{https://developer.intel.com/quantumsdk}}~\citep{Khalate2022}, including the automated remapping of measurement results.

The structure of the paper is as follows: In Section~\ref{sec:motive}, we motivate the set of optimization implemented by \PCAST, before giving some theoretical background in Section~\ref{sec:back}. We then discuss the Stabilizer Search Problem and our solution to it in Section~\ref{sec:stab} as well as the optimization of unitary channels in the presence of preparations in Section~\ref{sec:prepredux}. Section~\ref{sec:fullopt} combines these into a full optimization scheme. In Section~\ref{sec:eval}, we evaluate the performance of our methods for the Variational Quantum Eigensolver (VQE) algorithm where we find an overall average ratio of the maximum theoretical number of two-qubit gates to the actual number of two-qubit gates used by our method to be $7.91$. This is compared to the value of $3.5$ as found for a similar metric using the best known existing solution to the Stabilizer Search problem\cite{Crawford2021}. Finally, we make concluding remarks in Section~\ref{sec:conc}.

\section{Motivation} \label{sec:motive}
\begin{figure*}
\centering
\subfloat[Positive control on a qubit prepared in the $\ket{0}$ state is equivalent to just the preparation.]{
    \begin{tikzpicture} \node[scale=0.7] at (0,0) {
      \begin{quantikz}[row sep=0.1cm]
        & \gate{\PrepZ} & \ctrl{1} & \qw\\
        & \qw & \gate{\RX{\theta}} & \qw
      \end{quantikz}};
      
      \node at (2, 0) {\Large =};

     \node[scale=0.7] at (3.25, .27){
      \begin{quantikz}[row sep=0.1cm]
        & \gate{\PrepZ} & \qw
      \end{quantikz}};

    \node at (5.25, 0) {\Large $\Rightarrow$};

    \node[scale=0.7] at (8,0) {
       \begin{tikzcd}[row sep=0.2cm]
          \startingpoprgate{\Prep{Z_0}{X_0}} \arrow[r] 
            &\poprgate{\Rot{Z_0 X_1}{-\theta/2}}
          \\
          &\startingpoprgate{\Rot{X_1}{\theta/2}}
      \end{tikzcd}
      };
      
      \node at (10.25, 0) {\Large =};
      
      \node[scale=0.7] at (11.5,0) {
        \begin{tikzcd}[row sep=0.2cm]
          \poprgate{\Prep{Z_0}{X_0}}
        \end{tikzcd}
      };
    
     \end{tikzpicture}
    \label{fig:motivea}
  }

  \subfloat[Negative control on a qubit prepared in the $\ket{0}$ state is equivalent to just the preparation and the uncontrolled unitary. ]{
    \begin{tikzpicture} \node[scale=0.7] at (0,0) {
      \begin{quantikz}[row sep=0.1cm]
        & \gate{\PrepZ} & \octrl{1} & \qw\\
        & \qw & \gate{\RX{\theta}} & \qw
      \end{quantikz}};
      
      \node at (2, 0) {\Large =};

     \node[scale=0.7] at (3.25, 0){
      \begin{quantikz}[row sep=0.1cm]
        & \gate{\PrepZ} & \qw \\
        & \gate{\RX{\theta}} & \qw
      \end{quantikz}};

    \node at (5.25, 0) {\Large $\Rightarrow$};

    \node[scale=0.7] at (8,0) {
       \begin{tikzcd}[row sep=0.2cm]
          \startingpoprgate{\Prep{Z_0}{X_0}} \arrow[r] 
            &\poprgate{\Rot{Z_0 X_1}{\theta/2}}
          \\
          &\startingpoprgate{\Rot{X_1}{\theta/2}}
      \end{tikzcd}
      };
      
      \node at (10.25, 0) {\Large =};
      
      \node[scale=0.7] at (11.5,0) {
        \begin{tikzcd}[row sep=0.2cm]
          \poprgate{\Prep{Z_0}{X_0}}\\
          \startingpoprgate{\Rot{X_1}{\theta}}
        \end{tikzcd}
      };
    
     \end{tikzpicture}
    \label{fig:motiveb}
    
}

\subfloat[Equivalent example as Fig.~\ref{fig:motivea} but in this case, the unitary is a Clifford, represented by a Pauli frame (see Sec~\ref{sec:frame}).]{
    \begin{tikzpicture} \node[scale=0.6] at (0,0) {
      \begin{quantikz}[row sep=0.1cm]
        & \gate{\PrepZ} & \ctrl{1} & \qw\\
        & \qw & \ctrl{} & \qw
      \end{quantikz}};
      
      \node at (2, 0) {\Large =};

     \node[scale=0.6] at (3.25, .14){
      \begin{quantikz}[row sep=0.1cm]
        & \gate{\PrepZ} & \qw
      \end{quantikz}};

    \node at (5.25, 0) {\Large $\Rightarrow$};

    \node[scale=0.6] at (8,0) {
        \begin{tikzcd}[row sep=0.2cm]
          \startingpoprgate{\Prep{Z_0}{X_0}} \arrow[r]
          & \startingframegate{1}{\begin{pmatrix}
              Z_0 & X_0 Z_1 \\
              Z_1 & Z_0 X_1
            \end{pmatrix}}
        \end{tikzcd}
      };
      
      \node at (10.25, 0) {\Large =};
      
      \node[scale=0.6] at (12.5,0) {
        \begin{tikzcd}[row sep=0.2cm]
          \poprgate{\Prep{Z_0}{X_0}}
          & \startingframegate{1}{\begin{pmatrix}
                Z_0 & X_0 \\
                Z_1 & X_1
          \end{pmatrix}}
        \end{tikzcd}
      };
    
     \end{tikzpicture}
    \label{fig:motivec}
  }

  \subfloat[Similar to Fig.~\ref{fig:motivec} but here we demonstrate the resulting \PCAST interpretation on downstream measurements for a release outcome.]{
    \begin{tikzpicture} \node[scale=0.6] at (0,0) {
      \begin{quantikz}[row sep=0.1cm]
        & \gate{\PrepZ} & \ctrl{1} & \gate{\MeasX}\\
        & \qw & \targ{} & \gate{\MeasZ}
      \end{quantikz}};
      
      \node at (2.25, 0) {\Large =};

     \node[scale=0.6] at (4, 0){
      \begin{quantikz}[row sep=0.1cm]
        & \gate{\PrepZ} & \gate{\MeasX}\\
        & \qw & \gate{\MeasZ} &
      \end{quantikz}};

    \node at (6.25, 0) {\Large $\Rightarrow$};

    \node[scale=0.6] at (9.25, 0) {
       \begin{tikzcd}[row sep=0.2cm]
         \poprgate{\Prep{Z_0}{X_0}} \arrow[r] \arrow[rd]
         &\startingpoprgate{\Meas[c_0]{X_0 X_1}}\\
         &\startingpoprgate{\Meas[c_1]{Z_0 Z_1}}
       \end{tikzcd}
     };

    \node at (12, 0) {\Large =};

    \node[scale=0.6] at (14.25, 0) {
       \begin{tikzcd}[row sep=0.2cm]
         \poprgate{\Prep{Z_0}{X_0}} \arrow[r]
         &\startingpoprgate{\Meas[c_0]{X_0}}\\
         &\startingpoprgate{\Meas[c_1]{Z_1}}
       \end{tikzcd}
     };
    
     \end{tikzpicture}
    \label{fig:motived}

  }

\caption{
Simple optimizations following preparations, and the corresponding optimizations on \PCAST graphs.
} \label{fig:motive}
  
\end{figure*}

Much of \PCAST's efficacy comes from the structure of the \PCAST graph, primarily by exploiting commutativity of the operations. This allows for merging of some operations and synthesizing of a more optimal ordering when implemented. However, this alone is not enough to capture many known simplifications, especially at the interfaces of unitary and non-unitary elements.
Fig.~\ref{fig:motive} shows some common quantum circuit optimizations involving preparations. Each of these is an example of a unitary controlled on a qubit prepared in a computational state and is not captured by the structure of the \PCAST graph on its own.
Furthermore, the result of each example translates to the \PCAST graph differently, and optimizations which capture and generalize these have to be handled separately. This is not a drawback, however, because the generalization can express otherwise non-intuitive relations that go well beyond simple pattern matching. We demonstrate this process in Section~\ref{sec:prepredux}. 


Another opportunity for optimization using \PCAST comes from a mutually commuting set of measurement nodes occurring at the end of a \PCAST graph. Such a case occurs when the graph is built from a quantum circuit that ends by measuring most or all of its qubits. A similar situation occurs when measuring an arbitrary Hamiltonian for the broad class of VQE algorithms \cite{Peruzzo_2014, Bharti2022, Fedorov2022} or any other class which requires the expectation value of a Hermitian observable. Though other methods exist \cite{Huggins2021, Klymko2021, Huang2020}, a common strategy for extracting expectation values is to expand the operator in the Pauli operator basis, then group the terms of that expansion into mutually commuting sets. These sets can either be qubit-wise commuting \cite{Verteletskyi2020}, fully commuting \cite{Gokhale2020} or sorted for the sake of minimizing total shots on the quantum computer \cite{Crawford2021}. The qubit-wise commutation case has a direct translation to a circuit---a single-qubit basis change to the agreed single-qubit Pauli operator support (see Defn.~\ref{def:agree}). In the latter two cases, however, the authors generate a  Clifford circuit to extract a set of simultaneous measurements capable of reconstructing all elements. (see Defn.~\ref{def:stabsearch} for a formal description). 


\begin{figure}
\centering

  \begin{tikzpicture}

    \node[scale=0.6] at (0,0) {
       \begin{tikzcd}[row sep=0.2cm]
          \startingpoprgate{\Meas[c_0]{Z_0}} &
          \startingpoprgate{\Meas[c_1]{Z_0 Z_1}} &
          \startingpoprgate{\Meas[c_2]{Z_0 Z_1 Z_2}} \\   
      \end{tikzcd}
      };
      
      \node at (0, -.5) {\Large $\Downarrow{}$};
      
      \node[scale=0.6] at (0,-1.3) {
      \begin{quantikz}[row sep=0.3cm]
         \qw & \gate{\MeasZ^{c_0}} & \targ{} & \gate{\MeasZ^{c_1}} & \targ{} & \gate{\MeasZ^{c_2}} & \targ{} & \targ{} & \qw \\
         \qw & \qw & \ctrl{-1} & \qw & \qw & \qw & \ctrl{-1} & \qw & \qw\\
         \qw & \qw & \qw & \qw & \ctrl{-2} & \qw & \qw & \ctrl{-2} & \qw
        \end{quantikz}
      };

    \end{tikzpicture}
\caption{Example of how PCOAST circuit synthesis may handle measurements because it must meet an end-state guarantee implied by the PCOAST graph representation.} \label{fig:oddsynth}
\end{figure}

When applying synthesis to a \PCAST graph with mutually commuting Pauli measurements at the end, the greedy synthesis algorithm described in Ref.~\companion can generate seemingly odd behavior, as shown in Fig.~\ref{fig:oddsynth}. Here, we see that several measurements are performed on the same qubit, and CNOT gates are used to perform a binary sum over classical variables which could be performed by classical resources. Moreover, gate action after the measurements may cause problems if the measurements on the underlying qubit system are destructive. This behavior cannot be remedied by simply replacing the sequence of measurement and Clifford gates with some equivalent sequence as described by the Refs.~\cite{Verteletskyi2020,Gokhale2020, Crawford2021}, as there is an implicit promise that the resulting quantum channel action is fully equivalent to what was originally specified, including the end quantum state returned by the channel. We refer to this as the \emph{end-state promise}.   
The end-state promise even applies when the qubit states are fully measured and entanglement is completely removed. Baring some classical feedback mechanism, synthesis does not know a priori what the outcomes will be and has no choice but to use Clifford gates to 
reconstruct the classical state and meet the promise.
However, there are many situations where the end-state is irrelevant. 
This leaves two possible desired outcomes:

\begin{enumerate}
    \item A \emph{hold outcome} indicates that the quantum state is in part or in whole a desired return of the program. Thus we must meet the end-state promise.
    \item A \emph{release outcome} indicates that the desired results are the measurement outcomes only. Once all measurement results have been recorded, the quantum state can be ``released'' and we make no promises on the final state of the quantum system. More formally:
\end{enumerate}
    \begin{definition}\label{def:release}
    A release outcome is represented by an equivalence class defined by a set of measurables, $ M = \{m_i\}$, of quantum channels such that $C$ is equivalent to $C'$ relative to $M$ if and only if $\tr(C(\rho) m_i) = \tr(C'(\rho) m_i)$ for each $i$ and for all states $\rho$.\footnote{This definition is equivalent to that provided in \companion, as written in terms of \emph{classical-quantum states}, where $M$ corresponds the space of measurement outcomes described therein.} 
    \end{definition}
    In the case of a \PCAST graph, $M$ is represented by the measurement nodes of the graph, which is implied when discussing a release outcome.

For most algorithms or quantum submodules, a release outcome is sufficient, including variational algorithms such as VQE, Thermofield double state preparation~\cite{Wu2019, Zhu2020}, and imaginary-time evolution~\cite{McArdle2019, Motta2020}, to name a few.
Hold outcomes are typically only required for submodules where partial measurements are used to update the state in real-time, such as syndrome extraction for error correction or repeat-until-success routines~\cite{Paetznick2014, Moreira2022}.

In the case of a release outcome, we are free to replace a set of mutually commuting measurements with an equivalent, more efficient set of measurements. In Section~\ref{sec:stabsearch}, we propose a method for simultaneously deriving both an equivalent set of measurements as well as a circuit realizing them. 


\section{Background}\label{sec:back}
\subsection{The Pauli group}

A single-qubit Pauli is one of $X$, $Y$, $Z$, defined as the $2\times 2$ matrices,
\begin{align}
    X &= \begin{pmatrix}
        0 & 1 \\
        1 & 0
    \end{pmatrix} \qquad
    Y = \begin{pmatrix}
        0 & -i \\
        i & 0
    \end{pmatrix} \qquad
    Z = \begin{pmatrix}
        1 & 0 \\
        0 & -1
    \end{pmatrix},
\end{align}
or the identity $I$. At times we use $\sigma \in \{X, Y, Z\}$ to represent an arbitrary, non-identity single-qubit Pauli. A Pauli operator, $P$, on $N$ qubits is then the tensor product of any combination of single-qubit Paulis, scaled by a constant $\text{phase}(P) \in \{1,-1,i,-i\} \simeq \mb{Z}_4$. By convention, the canonically positive version of a Pauli operator is the one which can be written as such with the phase $1 \in \mb Z_4$. This collection of operators can be understood as the \emph{Pauli group}, $ \mathcal{G}_N$, under matrix multiplication. The support of an $N$-qubit Pauli operator, $\support(P)$, is the set of indices $i$ for which $(P)_i \neq I$. We write $X_i$, $Y_i$, and $Z_i$ for Paulis with support $\{i\}$.

An important feature of the Pauli group is that any two members $P_1, P_2 \in \mc{G}_N$ either commute or anti-commute. We capture this in the function $\lambda: \mc{G}_N \times \mc{G}_N \to \mb{F}_2$, 
\begin{align}
    \lambda(P_1,P_2) &= \begin{cases}
        0 & P_1 \text{ commutes with } P_2 \\
        1 & \text{otherwise}
    \end{cases}
\end{align}
We have, therefore, that $P_1 \cdot P_2 = (-1)^{\lambda(P_1,P_2)} P_2 \cdot P_1$.

\subsection{The Pauli Space}
Throughout this paper, the phase factor of the Pauli group, $\mb{Z}_4$, tends to function as a hindrance to the understand of our optimizations and their functioning. This is the case because of the overwhelming use of the Hermitian versions of Pauli operators (i.e. when $\mb Z_4$ is restricted to $\{1, -1\} \simeq \mb Z_2$) and the ability for the distinction between the canonically positive and negative version of a Pauli operator to be absorbed by other mechanisms. As such, we define the \emph{Pauli space} as $\mc P_N = \mc G_N / \mb Z_4$. The Pauli space abelianizes the Pauli group, and thus can be understood as a binary vector space $\mc P_N \simeq \mb Z_2^{\oplus 2N}$\cite{Calderbank1997}. From this perspective, multiplication is lifted to binary addition, implying $I$ is the additive identity, multiplicative power is lifted to scalar multiplication and the binary field $\mb F_2$ is appropriate since $P^2 \propto I$ in $\mc G_N$. An overview of this mapping is show in Table~\ref{fig:GtoP}. We denote the lift of a Pauli operator $P \in \mc G_N$ to its $p \in \mc P_N$ equivalence class via square brackets, $p =[P]$, and implicitly always map members of $\mc P_N$ to their canonically positive element in $\mc G_N$. Although we lose the commutation relations, we can reintroduce them by extending $\mc P_N$ to a \emph{symplectic} vector space, by treating $\lambda$ as a symplectic form (as restricted to $\mc P_N$).\footnote{The necessary feature of a symplectic form $\lambda$ are for $p,q,r \in \mc P_N$:
$\lambda(p + q, r) = \lambda(p, r) + \lambda(q, r)$,
$\lambda(p, p) =  0$, and 
$p = 0$ if and only if $\lambda(p,q) = 0$ for all $q \in \mc P_N$. Note that these conditions imply $\lambda(p,q) = \lambda(q,p)$ in the binary case. Also for a subset $\mc G \subseteq  \mc P_N$, we define $\mc G^{\perp} = \{p \in \mc P_N : \lambda(p, \mc G) = \{0\}\}$.}
\begin{table}
    \caption{Overview of the correspondence between the Pauli Group and the Pauli Space. The use and interpretation of this notation should be clear in the context in which it is used.}
    \centering
    \begin{tabular}{lll} \toprule
      \textbf{Notion} & \textbf{Pauli Group} & \textbf{Pauli Space}\\
      \midrule
      membership & $P \in \mc G_N$ & $p = [P] \in \mc P_N$ \\
      scalar action ($a \in \mathbb F_2$) & $P^a$ & $ap$ \\
      group action & $P * Q$ & $p + q$ \\
      identity & I & 0 \\
      2-form & $\lambda(P, Q)$ & $\lambda(p,q)$\\
      Clifford action & $\fwdAction F(P)= U P U^{\dagger}$ & $f(p)$ \\
    \bottomrule
\end{tabular}
    \label{fig:GtoP}
\end{table}

\subsection{Pauli Frame as Pauli Space Automorphism} \label{sec:frame}

Pauli tableaus~\cite{Aaronson2004} were first introduced as a way to simulate stabilizer states generated entirely from Clifford gates ($H$, $S$, and $\CNOT$) and single-qubit measurements. Since then, Pauli tableaus have been used to represent Clifford circuits in general, not just for the purposes of stabilizer simulation. Following \citet{Schmitz2021}, in this work we refer to Pauli tableaus as \emph{Pauli frames} to emphasize the linear algebraic structure they represent for the Pauli space. 
\begin{definition}
A Pauli frame $F$ on $N$ qubits is an $N \times 2$ array of Hermitian Pauli operators,
\begin{align}\label{eq:frame}
    F = \begin{pmatrix}
       \effZ_0 & \effX_0 \\
      \effZ_1 & \effX_1 \\
        \vdots & \vdots \\
        \effZ_{N-1} & \effX_{N-1}
    \end{pmatrix},
\end{align}
which satisfy the following commutation relations:
\begin{subequations} \label{eq:comrel}
   \begin{align}
        \commutativity{\effZ_i}{\effZ_j} =& \commutativity{\effX_i}{\effX_j}
        = 0 \\
        \commutativity{\effZ_i}{\effX_j} =
        & \delta_{ij}.
    \end{align}
\end{subequations}
We also define $\effY_i = -i \effZ_i * \effX_i$ and the \emph{origin} frame $F^0$ such that $\eff \sigma_i = \sigma_i$ for all $i$.
Furthermore, we say a Pauli $P \in \mc G_N$ is \emph{in the frame} ($\in$) $F$ if there exists an index $0\leq i < N$ such that $P = \effX_i$, $P = \effY_i$ or $P = \effZ_i$.
\end{definition}

An extensive discussion on the understanding of Pauli frames as Clifford unitaries can be found in \companion. Instead, we focus our discussion here on Pauli frame as a representation of a \emph{symplectic automorphism} on the Pauli space. A symplectic automorphism on $(\mc P_N, \lambda)$ is any linear map $f: \mc P_N \to \mc P_N$ which preserves the symplectic form.\footnote{For all $p,q \in \mc P_N$, 
 $\lambda( f(p), f(q)) = \lambda(p, q).$
Some immediate consequences of this definition and the properties of $\lambda$ are that $f^{-1}$ exists and 
 $\lambda(f(p), q) = \lambda(p, f^{-1}(q)).$
}
Given a Pauli frame as defined in Eq.\eqref{eq:frame}, one can lift it to two symplectic automorphism defined by
\begin{subequations}
\begin{align} \label{eq:framelift}
\fwdAction{F}(p) = \sum_i \left( \lambda(p, [\X_i]) [\effZ_i] + \lambda(p, [Z_i]) [\effX_i] \right), \\
\bwdAction{F}(p) = \sum_i \left( \lambda(p, [\effX_i]) [Z_i] + \lambda(p, [\effZ_i]) [X_i] \right). 
\end{align}
\end{subequations}
Furthermore, we use the term \emph{frame}\cite{Bernhard2009} as it emphasizes the following linear algebraic properties. It should be obvious from the commutation relations $F$ satisfies that $\fwdAction F([Z_i]) = [\effZ_i]$ and $\fwdAction F([(X_i]) = [\effX_i]$. $F^0$ is clearly lifted to a basis for $\mc P_N$, which implies that elements of $F$ are also lifted to a basis via the expansion of any $p \in \mc P_N$,
\begin{align} \label{eq:framelift}
p = \sum_i \left( \lambda(p, [\effX_i]) [\effZ_i] + \lambda(p, [\effZ_i]) [\effX_i] \right).
\end{align}
So we can interpret a Pauli frame in this context as both a basis for $\mc P_N$ as a symplectic frame as well as a symplectic automorphism. From this and Eq. \eqref{eq:framelift}, it is clear that $\bwdAction F = \left(\fwdAction F\right)^{-1}$.\footnote{This provides a means for inverting a signed Pauli frame and thus the Clifford unitary it represents by asserting the signs of $F^{-1}$ are such that $\fwdAction F(F^{-1}) = F^0$.}
It should also be clear that any symplectic automorphism, $f$ can be interpreted as a canonically positive Pauli frame via $F = \{(f([Z_i]), 
 f([X_i])\}_{0\leq i < N}$ and thus can be mapped to a Clifford unitary.\footnote{More accurately, the set of sympletic automorphisms is isomorphic to the Clifford group quotient the Pauli group, and we are arbitrarily using canonical positivity of the elements as a means of select an element from each equivalence class of the quotient group.} It can also be shown that any set of $2N$ members of $\mc P_N$ when divided into two groups 
forms a symplectic frame, i.e. can be mapped to a Clifford unitary, if and only if they satisfy the commutation relations in Eqs.\eqref{eq:comrel}.
Importantly, linear independence in such a set is implied by the commutation relations. In fact, a corollary to this result is any set of such pairs of elements of $\mc P_N$ less than $N$ must also be linearly independent.
\subsection{The \PCAST Graph and its Semantics} \label{sec:pcoast}
The primary representation which we consider optimizing is the \emph{\PCAST graph}. In particular, we assume a \emph{frame-terminating, fully-merged} \PCAST graph as described in Ref. \companion. We leave a detailed discussion of the process of converting to the \PCAST graph representation and synthesis back to a circuit representation to the reference, but give a brief overview of the structure of the \PCAST graph representation here.

A \PCAST graph is a representation of a generic quantum program which can be described as a sequence of gate-like channels represented as nodes. We use the notation $\interp{\cdot}$ when interpreting a formal node or collection of them called a \emph{term} as a channel, where
\begin{align}
\interp{t_1; t_2} = \interp{t_2}\circ \interp{t_1}.
\end{align}
The direct equivalence between interpreted terms is referred to as \emph{hold equivalence}, $\equiv^{\text{hold}}$, which is contrasted with \emph{release equivalence}, denoted $\equiv^{\text{release}}$, when the interpreted terms both satisfy Defn.~\ref{def:release}.
The sequencing of constituent nodes of a term is then encoded in a directed acyclic graph (DAG) where nodes are non-Clifford unitary and non-unitary quantum channels. A \emph{singlet node} is any node defined by a single Pauli operator $P$ as represented by $n(P)$. This includes $\Rot{P}{\theta}$ which represents a Pauli rotation around the axis $P$ by an angle of $\theta/2$ and $\Meas[c]{P}$ which represents a projective measurement of the eigenvalues of $P$ where the outcome is recorded to classical variable $c$. A \emph{factor node} is any node defined by a pair of non-commuting Pauli operators $(P, Q)$ as represented by $n(P,Q)$. This includes $\Prep{P}{Q}$ which represents a preparation or reset of the state to the $+1$ eigenspace of $P$ using $Q$ as conditional operation to take the state out of the $-1$ eigenspace. Other single qubit rotations can be generalized as factor nodes, but we avoid discussing them here for clarity. Non-classical nodes are also included for mapping measurement outcomes.

Edges are drawn in the graph based on non-commutativity of the defining Pauli operators with the direction of the edge determined by the logical order of operation for the quantum channel it represents. Thus any topological sort of the DAG results in the same overall channel. The assumption of being fully-merged means all nodes satisfying the merging rules found in \companion are merged so long as they are not path-connected or \emph{incomparable}. A \PCAST graph is frame terminating when it contains a single Pauli frame node representing a residual Clifford unitary channel to be applied last. As discussed in \companion, any \PCAST graph can be reduced to this form and the semantics of frames is such that for any Pauli operator $P \in \mc G_N$ and terminating frame $F$,
\begin{align}
\interp{U^F}(P) = \interp{F}(P)= \bwdAction F(P).
\end{align} 
This also implies that 
\begin{align}
\interp{F_1; F_2} = \interp{F_1 \circ F_2}
\end{align}
where we define $F_1 \circ F_2 = \fwdAction F_1 \circ \fwdAction F_2(F_0) = \fwdAction F_1(F_2)$.\footnote{The action is implied to be entry-wise.} 

\section{Stabilizer Search}\label{sec:stab}
\label{sec:stabsearch}
Several of the optimizations we describe below depend on an algorithmic solution to the following problem:
\begin{definition} \label{def:stabsearch}
(\emph{Stabilizer search problem}) Let $S$ be a mutually-commuting (Hermitian) subset of $\mc G_N - \mathbb Z_2$. The span of $S$ forms the \emph{stabilizer subspace}, $\Span(S) = \mc G_S \subset \mc P_N$. Given some cost metric on Pauli frames (namely cost to implement as a circuit), find a low cost frame, $F$, for which there exists a similar mutually commuting set $S'$ such that $S' \subset F$ and $\mc G_S \subseteq \Span(S') = \mc G_{S'}$. If we require $\mc G_S = \mc G_{S'}$, we refer to it as the \emph{exact} stabilizer search problem.
\end{definition}

The primary use case for this problem is to measure a set of mutually commuting Pauli operators using the fewest resources as discussed in Section~\ref{sec:motive}. Deriving our solution to this problem is an interesting exercise in leveraging the symplectic vector space formalism. For brevity, we provide an overview of the key insights from such a derivation along with a proof of its correctness as a solution.

The power of couching the problem in the language of linear algebra is that we have a well-understood analog from which to draw intuition. As we generally expect the initial set $S$ to over-determine its span, we can see this problem as a variant of finding the null space or row space of a non-symmetric matrix. In particular, we characterize the null space of the \emph{syndrome map} in the general case or the row space of the \emph{stabilizer map} in the exact case, using the language of quantum error correction\cite{Gottesman1997, Schmitz2019}. Either way, the solution invokes Gaussian reduction, but with the analogous set of row operations, noting the \emph{two-qubit entangling} (TQE) gates are as outlined in Ref.~\cite{Schmitz2021}:\footnote{For a given qubit pair there are $3 \times 3 = 9$ TQE gates, generalizing CNOT and CZ, such that we have one Pauli operator basis for each qubit.}
\begin{enumerate}
\item Row swapping $\to$ virtual qubit/Pauli position swapping,
\item Row scalar multiplication $\to$ nothing as the field is binary,
\item In-place row addition $\to$ applying a TQE gate.
\end{enumerate}
Thus the algorithm, by virtue of (3), simultaneously produces a circuit which preforms the Clifford transformation. We also note that the process of reducing to row echelon form has the property that the column of the pivot for a reduced row can be effectively ignored from then on. This is analogously handled here by \emph{support masking} provided by the function $\mask : \mc G_N \times (\text{subsets of qubits}) \to \mc G_N$, where
\begin{align}\label{eq:mask}
    \mask(P, \Q) = \text{phase}(P) \prod_{i \in \Q} (P)_i.
\end{align}
That is, it reduces the support of $P$ to that which is contained in $\Q$. One Pauli-space specific consideration is that of \emph{support agreement}: 
\begin{definition} \label{def:agree}
For any mutually commuting set $S \subset \mc G_N$, we say $S$ agrees on support $\sigma = X, Y, Z$ for qubit $i$ iff for every element $s \in S$, $\lambda(s, \sigma_i) = 0$. That is, every element of $S$ either has support $\sigma$ or no support on $i$.
\end{definition}
With this concept, we can paraphrase Defn.~\ref{def:stabsearch} as the transformation of $S$ by a Clifford unitary such that every element agrees on support for every qubit. Thus we propose the \emph{Stabilizer Search Algorithm} in Fig.~\ref{algo:stab}, as described for the use of simultaneous measurements.

\begin{figure}
    \begin{algorithmic}[1]
    \Function{FindStabilizers}{$S$, ``general'' or ``exact''}
        \State $S'' \gets S$, $\Q \gets \support(S'')$, $F' \gets F^0$, $C \gets \emptyset$.
        \If{this is the ``general'' case}
          \For{$i \in \Q$} \label{line:genstart}
            \If{$S''$ agrees on support $\sigma$ for qubit $i$}
              \State Add a measurement for $\sigma_i$ to $C$.
              \State Remove $i$ from $\Q$.
            \EndIf
          \EndFor \label{line:genend}
        \EndIf
        \While {$\Q$ is not empty}
            \For{each $s \in S''$}
                \If{$\support(\mask(s, \Q))$ is empty}
                \State Remove $s$ from $S''$
                \ElsIf{$\mask(s, \Q) = \pm \sigma_i$ for some $i\in \Q$} \label{line:found}
                     \State Add measurement of $\pm \sigma_i$ to $C$.
                     \State Remove $i$ from $\Q$.
                \EndIf
            \EndFor
            \State Min $\gets \{\argmin_{s \in S''} \support(\mask(s, \Q))\}$ \label{line:gatebegin}
            \State MinGate $\gets \{ \text{TQE gates which reduces } \support(\text{Min}) \}$ 
            \State $g_{\text{min}} \gets \argmin_{g \in \text{MinGate}} \text{Cost}(g)$ \label{line:gateend}
            \State Add $g_{\text{min}}$ to $C$.
            \State $S'' \gets \fwdAction F_{g_{\text{min}}^{-1}}(S'')$,  $F'\gets F_{g_{\text{min}}^{-1}} \circ F'$. \label{line:applygate}
        \EndWhile
        \State \Return $(C, F'^{-1})$.
    \EndFunction
    \end{algorithmic}
  \caption{Psuedocode for Stabilizer Search Algorithm. Cost$(g)$ is any cost metric on the use of the TQE gate $g$ for reducing support.} \label{algo:stab}
\end{figure}

We defer any discussion around termination of \textsc{FindStabilizers} to the implementation section below and take for granted that the algorithm terminates. To prove Fig.~\ref{algo:stab} satisfies Defn.~\ref{def:stabsearch}, we prove the following:
\begin{lemma}
For $F$ returned by \textsc{FindStabilizers} on set $S$, the set $\fwdAction F^{-1}(S)$ agrees on its support for all qubits.
\end{lemma}
\begin{proof}
In the general case, $S$ automatically agrees on all qubits $i$ which are removed in lines \ref{line:genstart}-\ref{line:genend} of Fig.~\ref{algo:stab}. Otherwise, all qubits are in $\Q$ at the start and $S$ agrees on all qubit not in $\Q$ vacuously. Then let $\Delta \Q_{\text{init}}$ contain all qubits not in $\Q$. 

For the sake of argument, suppose at step $m$ and for all $m' \leq m$ through the main loop of the algorithm, $\fwdAction F_{(m')}'(S) = S'_{(m')}$ agrees on it support for all qubits not in $\Q_{(m')}$, where $(m')$ denotes the $m'^{th}$ versions of $F'$, $\Q$ and other quantities at the end of that iteration. 
Now consider the beginning of the $(m+1)^{th}$ iteration. At step (1), if $\mask(s, \Q_{(m)})$, for $s \in S''_{(m)} \subseteq S'_{(m)}$, has no support, then it agrees with the rest of $S'_{(m)}$ by our hypothesis and is then absent from $S''_{(m+1)}$. If $\mask(s, \Q_{(m)}) = \pm \sigma_q$, then consider for any $s' \in S'_{(m)}$,
\begin{align}
\lambda(\sigma_q, s') =& \lambda_q(\mask(s, \Q_{(m)}), s') \nonumber \\
=& \sum_{i \in \Q_{(m)} } \lambda_i(s, s') +\sum_{i \notin \Q_{(m)} } \lambda_i(s, s') \nonumber \\
=& \lambda(s, s') = \lambda(\fwdAction F_{(m)}^{-1}(s), \fwdAction F_{(m)}^{-1}(s')) \nonumber \\
=& 0. \label{eq:agrees}
\end{align}
Eq.~\ref{eq:agrees}\footnote{The subscript on $\lambda$ represents its functional support on that qubit,  the first equality is a consequence of $s$ have single-qubit masked support, the second is a consequence of our induction hypothesis,  and the final two equalities are a consequence of the automorphism property of frames and the fact that $S$ is a mutually commuting set.}  implies that all of $S'_{(m)}$ agrees on the support of $\sigma$ on $q$ and as we remove it, $q \notin \Q_{(m+1)}$. Thus by the end of step (1), all changes to $\Q$ are made such that $\Q_{(m)} \to \Q_{(m+1)}$, and all of $S''_{(m)}$ agrees on support for all of $\Delta \Q_{(m)} =  \Q_{(m)} - \Q_{(m+1)}$. 

The TQE gate selection in lines \ref{line:gatebegin}-\ref{line:gateend} insures $g_{\text{min}(m)}$ only acts on elements of $\Q_{(m+1)}$. Therefore since $S'_{(m)}$ agree on support for qubits not in $\Q_{(m+1)}$, it is the case that $S'_{(m+1)} = \fwdAction F_{g_{(m)}^{-1}}(S'_{(m)})$ must also agree on support for qubits not in $\Q_{(m+1)}$.

Assuming the algorithm terminates, there exists $M$ such that $\Q = \bigcup_{0\leq m < M}\Delta \Q_{(m)} \bigcup \Delta \Q_{\text{init}}$ and $F = F_{(M)}^{-1}$.
Therefore by induction, we have that $\fwdAction F^{-1}(S)$ agrees on its support for all qubits.
\end{proof}

An immediate corollary to this proof is that every measurement in the circuit corresponds to the support agreement for all qubits $\Q$ on which $S$ has support. Let $\Sigma = \{\sigma_q\}_{ q \in \Q}$ be the set of support agreement as found from the returned circuit. As each element of $\Sigma$ is a single qubit Pauli operator on different qubits, every element of $S' = \fwdAction F(\Sigma)$ belongs to $F$ and is mutually commuting. To prove this set of operators satisfy Defn.~\ref{def:stabsearch}, it is sufficient to prove $S' \subseteq \mc G_S^\perp$. \footnote{This on its own is actually not sufficient. We need to also show that all members of $S$, when expanded in the basis $F$, contains only terms from elements of $F$ associated with $\Q$. However, it should be clear that this is trivially true as $S$ has no support outside of $\Q$ and the gates used to generate $F$ and thus by extension $F$ itself have no action on operators with no support in $\Q$.} But this immediately follows from the above proof since,
\begin{align}
\lambda(S, S') = \lambda(S, \fwdAction F(\Sigma)) = \lambda(\fwdAction F^{-1}(S), \Sigma) = \{0\},
\end{align}

Because $\fwdAction F(\Sigma)$ is contained in the frame $F$, we can use Eq.~\ref{eq:framelift} to expand members of $S$ as a product of elements in $S'$ via the \emph{Map Measurement Algorithm} as described in Fig.~\ref{algo:mm}.\footnote{We must address why the exact version of Fig.~\ref{algo:stab} satisfies the exact condition of Defn.~\ref{def:stabsearch}. The binary matrix $\underbar b$ returned by Algorithm~\ref{algo:mm}, when restricted to elements of $S$ satisfying line~\ref{line:found} of Fig.~\ref{algo:stab} has a upper triangular form and is thus inevitable.}

\begin{figure}
    \begin{algorithmic}[1]
      \Function{MapMeasurements}{$S, C, F$}
        \State $\Sigma \gets \{\sigma_q\}$, where each $\sigma_q$ is a measurement in $C$.
        \State $S' \gets \bwdAction F(\Sigma)$, $\underbar b \gets \mathbf{0}$, $v \gets \vec{0}$.
        \ForAll{$s_i \in S$}
          \State Pauli $P \gets I$
          \ForAll{$\sigma_q \in \Sigma$}
            \If{$\lambda(s_i, F(\tilde \sigma_q))$}
              \State $P \gets P * F(\sigma_q)$, $\underbar b_{iq} \gets 1$.
            \EndIf
          \EndFor
          \If{$\text{phase}(P) \neq \text{phase}(s_i)$}
            \State $v_i \gets 1$.
          \EndIf
        \EndFor
        \State \Return $(\underbar b, v)$
      \EndFunction
    \end{algorithmic}
    \caption{Psuedocode for Map Measurement Algorithm.}\label{algo:mm}
\end{figure}

\subsection{Implementation of the Stabilizer Search Algorithm}
It should be clear from the discussion in the Section V.C of Ref.~\companion that, given the appropriate \emph{search functions}, the greedy search algorithm for circuit synthesis is well suited for implementing Fig.~\ref{algo:stab}. As argued in \companion, we are then guaranteed the algorithm terminates. Moreover we can modify almost any version of the search functions used for circuit synthesis so long as they make an additional promise.\footnote{$\reduceNode$ as described in \companion only returns TQE gates entirely supported on the qubit support of the passed node.} Thus we introduce the \emph{stabilizer search template}. This is a templated implementation of the search functions which holds internally a list of qubits for the sake of masking and removes a qubit when a measurement is added to the circuit.

The complexity of Fig.~\ref{algo:stab} also follows from Ref.~\companion as $\mathcal O(N^3 k^2)$ where $k$ is the size of the outcome stabilizer set $S'$ which is also its stabilizer space dimension. The number of TQE gates it produces is at most $Nk - k(k+1)/2$, but tends toward far lower numbers as we demonstrate in Section~\ref{sec:eval} while also taking into consideration circuit depth and the target gate set.

\section{Unitary Reduction by Preparations}\label{sec:prepredux}
\subsection{Frame Reduction}
In this section, we discuss an algorithmic solution to the following problem: \begin{definition}
(\emph{Preparation-on-frame reduction problem})
Suppose we have a set of mutually commuting, un-equal preparation channels $\pprep(\{(P_i, Q_i)\}) = \pprep(\Pi)$, followed by a Clifford unitary channel as defined by the Pauli frame $F$. Given some cost metric on Pauli frames (again, typically circuit cost to implement), find a lower cost frame $F'$ and possibly lower cost representation $\Pi' = \{(P_i', Q_i')\} \simeq \Pi$ such that,
\begin{align}\label{eq:pfredux}
\interp{\pprep(\Pi) ; F} = \interp{\pprep(\Pi') ; F'}.
\end{align}
\end{definition}

 For $\pprep(P,Q)$, we refer to $P$ as the \emph{stabilizer}, which is privileged over $Q$ which we refer to as the \emph{destabilizer}\cite{Aaronson2004}. To understand the preparation-on-frame reduction problem, we first consider what is meant by $\Pi \simeq \Pi'$. As discussed in \companion, a pair of anti-commuting Pauli operators can always be understood as defining an effective qubit factorization of the Hilbert space. From the description of $\pprep$ in Section \ref{sec:pcoast}, we find that $\Pi$ generalize this such that it defines an effective multi-qubit factorization of the Hilbert space where the outcome of $\pprep(\Pi)$ is a simultaneous $+1$ eigenstate of the stabilizers, using a specific combination of destabilizers to transform out of any other alternative stabilizer subspace.\footnote{Alternatively, we understand equivalence  by considering a symplectic subspace of $\mc P_N$ with symplectic form provided by $\lambda$ as restricted to that subspace, for which $\Pi$ is a frame for that subspace and can be used as such. $\Pi'$ is then just an alternate frame for that subspace.} Thus $\pprep(\Pi')$ must have the same action on any state. Finding such an alternative should already evoke the exact version of Defn.~\ref{def:stabsearch}. Let $S = \stabs(\Pi) = \{P_i\}$, on which we can apply the exact version of Fig.~\ref{algo:stab} to find $S' = \stabs(\Pi')$ and $F$. To construct corresponding destabilizers, $\{Q_i'\}$, note they must satisfy $\lambda(P_i', Q_i') = \delta_{ij}$ and $\lambda(Q_i', Q_j') = 0$. For simplicity, assume $Q_i' \in \Span(\{Q_j\})$ such that the second condition is satisfied. Let $\tilde{Q}_i \in F$ be one of the two elements which anti-commutes with $P_i'$. We then find that,\footnote{Note that we can freely ignore sign here, because the sign of a destabilizer does not change the channel action.}
\begin{align}\label{eq:destconst}
[Q_i'] 
       =& \sum_j \lambda(Q_i', P_j) [Q_j] = \sum_{jk} \lambda(Q_i', P_k') \lambda(\tilde{Q}_k, P_j) [Q_j] \nonumber \\
       =& \sum_{j} \lambda(\tilde{Q}_i, P_j) [Q_j],
\end{align}

We also consider the equivalence in Eq. \eqref{eq:pfredux}. An alternative formulation of $\interp{\pprep(\Pi)}$ is as a partial trace over he subsystem/factor defined by $\Pi$, tensor product with the projection onto the appropriate stabilizer subspace, i.e.
\begin{align} \label{eq:prepchannel}
\interp{\pprep(\Pi)}(\rho) = \tr_{\Pi}(\rho) \otimes \proj(S),
\end{align}
where $\proj(S)$ is defined as the projection onto the simultaneous $+1$ stabilizer space of $S$. From this, it is clear that when $\rho$ is expanded in the Pauli operator basis, terms not wholly outside of the subsystem defined by $\Pi$ are zero under the partial trace and all remaining terms are multiplied in equal measure by elements from $\Span(S)$. Thus the frames $F$ and $F'$ in Eq.~\eqref{eq:pfredux} need only have the same action on this limited set of Pauli operators. This provides us with the following conditions:

\begin{lemma} \label{lemma:pfcond}
Eq.~\eqref{eq:pfredux} is true if and only if\footnote{Note, $s \in \Span(S)$, $s\proj(S) = \proj(S)$, i.e. the projection operator can always absorb a stabilizer. This is necessary for understanding Eq.\ref{eq:pfcondb}.}
\begin{subequations} \label{eq:pfcond}
\begin{align}
\bwdAction F(\Span(S))  \simeq & \bwdAction F'(\Span(S)), \label{eq:pfconda}\\
\frac{\bwdAction F(\Span(\Pi)^\perp \oplus \Span(S) )}{\Span(S)} =& \frac{\bwdAction F'(\Span(\Pi)^\perp \oplus \Span(S))}{\Span(S)} \label{eq:pfcondb}
\end{align}
\end{subequations}
\end{lemma}

The first says that $F$ can modify the stabilizers over what $F'$ does, but it must preserve the space.\footnote{Though we are using the language the Pauli space here, it is fundamental that the signs are also preserved by $F$ as well.} The second condition is more strict with the exact equality implying that any Pauli operator which commutes with all of the elements of $\Pi$ must be transformed exactly the same by $F$ and $F'$ modulo some element of $\Span(S)$.

This leads us to introduce a frame transformation we refer to as \emph{qubit decoupling}: 
\begin{definition}
For an anti-commuting pair $(P, Q)$ and frame $F$ as defined in Eq.~\ref{eq:frame} such that $P \in F$, the qubit decoupling transformation $\decouple(F, (P,Q))$ is defined as
\begin{align}
&\begin{pmatrix}
\effZ_0  P^{\lambda(\effZ_0, Q)}  & \effX_0  P^{\lambda(\effX_0, Q)} \\
\effZ_1  P^{\lambda(\effZ_1, Q)}  & \effX_1  P^{\lambda(\effX_1, Q)} \\
\vdots & \vdots \\
 P     &  Q  \\
 \vdots & \vdots \\
\effZ_{N-1}  P^{\lambda(\effZ_{N-1}, Q)}  & \effX_{N-1}  P^{\lambda(\effX_{N-1}, Q)} \\
    \end{pmatrix}.
\end{align}
\end{definition}
At times we omit the argument $(P,Q)$ for brevity. To show $\decouple$ results in a valid frame, note $P$ was a part of $F$ and has all the correct commutation relation with respect to the other elements of the resultant frame. Likewise is true for those elements multiplied by it. The only commutation relations of interest involve $Q$ for which,
\begin{align}
\lambda(Q, \effZ_i P^{\lambda(\effZ_i, Q)}) =& \lambda(Q, \effZ_i) +\lambda(P, Q)\lambda(\effZ_i, Q) \nonumber \\
=& 0  
\end{align}
Next we show the transform $\decouple$ preserves the properties of Lemma~\ref{lemma:pfcond}. Since $P$ is a member of the frame, Eq. \eqref{eq:pfconda} is satisfied trivially. Without loss of generality, suppose $P$ is $\effZ_q$ and consider $r \in \mc P_N$ such that $\lambda(P,r) = \lambda(Q, r) = 0$. Note this implies for $i \neq q$ $\lambda(r, \effZ_iP) = \lambda(r, \effZ_i)$, and likewise for $\effX_i$. Thus we find, 
\begin{align}
[\interp{\decouple(F)}(r)] =
[\interp{F}(r)] + \lambda(r, \effX_q) [P].
\end{align}
Therefore, $\decouple$ also satisfies Eq. \eqref{eq:pfcondb}.  

\begin{figure}
  \begin{algorithmic}[1]
     \Function{ReduceFrameByPrep}{$\Pi, F$}
     \State $S \gets \stabs(\Pi)$, $M \gets |S|$, $\tilde S \gets \bwdAction F(S)$
     \State $(C, F_{\text{aux}}) \gets \Call{FindStabilizers}{\tilde S, \text{``exact''}}$. \label{line:SSAcall}
     \State Let $\{\sigma_{q_i}\}_{0 \leq i < M}$ be the measurements in order from $C$.
     \State Construct $\Pi'$ using $\Pi$, $F \circ F_{\text{aux}}$ and Eq. \eqref{eq:destconst}.
     \State $c \gets \cost(\Pi)$, $c' \gets \cost(\Pi)$
     \If{$c' < c$} \label{line:branch}
       \State $\tilde F \gets \decouple(F\circ F_{\text{aux}}, \Pi')$ \Comment{in any order}
       \State \Return $(\Pi', \tilde F \circ F_{\text{aux}}^{-1})$.
     \Else \label{line:else}
       \State $\tilde F_0 \gets F\circ F_{\text{aux}}$.
       \For{$i= 0$ up to $M$}
         \State Let $(P_i, Q_i) \in \Pi$ such that $P_i = \bwdAction {\tilde F_i}(\sigma_{q_i})$.
         \State $\tilde F_{i+1} \gets \decouple(\tilde F_{i}, (P_i, Q_i))$
       \EndFor
       \State \Return $(\Pi, \tilde F_M \circ F_{\text{aux}}^{-1})$
     \EndIf
     \EndFunction 
  \end{algorithmic}
  \caption{Psuedocode for Preparation-on-Frame Reduction Algorithm. $\cost(\Pi)$ represents some cost evaluation for implementing of $\pprep(\Pi)$.} \label{algo:pfredux}
\end{figure}

$\decouple$ is a powerful tool for reducing the cost of a frame in the presence of a single $\pprep$, but requires the stabilizer to be a part of the frame. To generalize, we find an auxiliary frame, $F_{\text{aux}}$ such that $F \circ F_{\text{aux}}$ does contain the stabilizer. As such $F' = \decouple(F \circ F_{\text{aux}}) \circ F_{\text{aux}}^{-1}$, satisfies Eq.~\eqref{eq:pfcond}. This processes is then generalized  $\Pi$ or alternatively $\Pi'$. For this purpose, we introduce the \emph{Preparation-on-Frame Reduction Algorithm} in Fig.~\ref{algo:pfredux}. The correctness of Fig.~\ref{algo:pfredux} can then be argued with the following lemma:

\begin{lemma}
Every call to $\decouple$ in Fig.~\ref{algo:pfredux} is well-formed.
\end{lemma}

\begin{proof}
Suppose Fig.~\ref{algo:pfredux} enters the first branch at line~\ref{line:branch}. By construction, for each $\sigma_{q_i}$, $\bwdAction{F \circ F_{\text{aux}}}(\sigma_{q_i}) = P'_i$ is both in the frame $F\circ F_{\text{aux}}$, and a stabilizer of $\Pi'$. Also, the outcome of $\decouple(F \circ F_{\text{aux}}, (P_i', Q_i'))$ contains $P_j'$ for any $j \neq i$, since $\lambda(P_j', Q_i') = 0$. Therefore the decoupling sequence is well-formed.

Alternatively, suppose Fig.~\ref{algo:pfredux} enters the else branch at line~\ref{line:else}. Note, as a property of using of the exact version of Fig.~\ref{algo:stab}, we have that for all $m < M$
\begin{align} \label{eq:uppertri}
\bwdAction{\tilde F_0^{-1}}(P_m) = \sigma_{q_m} + \sum_{i < m} c_{mi} \sigma_{q_i},
\end{align}
where $c_{mi}$ is some binary scalar and $P_m$ is the 
$m^{th}$ element of $\stabs(\Pi)$ which was reduced to support 1 when masked in the call at line~\ref{line:SSAcall}.
In the first step of the loop, It is clear from Eq. \eqref{eq:uppertri}, that $\bwdAction{\tilde F_0}(\sigma_{q_0}) = P_0 \in \Pi$. Thus the first $\decouple$ is well-formed. 

For the sake of argument, now assume that for all $m'< m$, where $m < M$, the $m'^{th}$ $\decouple$ was well-formed. Now consider the $m^{th}$ step and $\bwdAction{\tilde F}_m(\sigma_{q_m}) \in \tilde{F}_m$ which, based on our hypothesis and the definition of $\decouple$ can be written as,
\begin{align}
\bwdAction{\tilde F}_m(\sigma_{q_m}) = \bwdAction{\tilde F}_0(\sigma_{q_m}) + \sum_{m' < m} \lambda(F_0(\sigma_{q_m}), Q_{m'}) P_{m'}
\end{align}
where we note that because $\bwdAction{\tilde F}_0(\sigma_{q_m}) \in \Span(\Pi)$, we can expand it in $\Pi$ as a basis, noting $\lambda(\bwdAction{\tilde F}_0(\sigma_{q_m}), P_{m'}) = 0$:
\begin{align}
\bwdAction{\tilde F}_0(\sigma_{q_m}) 
=& \sum_{m' < M} \lambda(\bwdAction{\tilde F}_0(\sigma_{q_m}), Q_{m'}) P_{m'} \nonumber \\
=& P_m + \sum_{m' < m} \lambda(\bwdAction{\tilde F}_0(\sigma_{q_m}), Q_{m'}) P_{m'},
\end{align}
where the last equality is a consequence of inverting the upper triangular form of Eq. \eqref{eq:uppertri} and the uniqueness of the expansion in this basis. From this it it immediately follows $\bwdAction{\tilde F}_m(\sigma_{q_m}) = P_m$ and thus the $\decouple$ for this step is well formed and by induction the sequence of $\decouple$s is well-formed.
\end{proof}

\subsection{Non-Clifford Reduction}
Non-Clifford nodes of the \PCAST graph can also be simplified when ``seen'' by a $\Prep{P}{Q}$. As demonstrated in Eq. \eqref{eq:prepchannel}, $P \sim I$ after the application of $\Prep{P}{Q}$. Thus for a singlet node dependent on the Pauli operator $R$, $n(R)$, which is incomparable with $\Prep{P}{Q}$, we have that
\begin{align}\label{eq:preponnode}
\interp{n(R);\Prep{P}{Q}} =& \interp{\Prep{P}{Q} ; n(R)} \nonumber \\
=& \interp{\Prep{P}{Q}; n(P*R)}.
\end{align}
Note this relation goes both ways. We can also consider applying this identity to a node defined by $R' = R*P$, i.e. $\lambda(R' P) = 0$ and $\lambda(R', Q) = 1$. Which version we choose for the sake of optimization depends on the cost of the outcome. We thus outline the use of Eq.~\eqref{eq:preponnode} for the \emph{Preparation-on-Node Reduction Algorithm} in  Fig.~\ref{algo:pnredux}.\footnote{Note we only check for node merging in line~\ref{line:mergable} and not earlier as to do so would be redundant by Eq.\eqref{eq:preponnode}.}

\begin{figure}
   \begin{algorithmic}[1]
     \Function{ReduceNodesByPrep}{$G$}
       \State $G' \gets G$, $\pprep(G') \gets$ all $\pprep$ nodes in $G'$.
       \For{$ n = n(P,Q) \in \pprep(G')$}
         \For{$n'(R)$ incomparable of $n$, $n' \notin \pprep(G')$}
           \If{$\cost(n'(R * P)) < \cost(n'(R))$}
             \State $R \gets R * P$ such that $n'$ comes after $n$.
           \EndIf
         \EndFor
         \For{$n(R)$ distance $+1$ from $n$, $\lambda(P, R) = 0$}
           \If{$\cost(n(R*P))< \cost(n(R))$}
             \State Replace $n(R)$ with $n(R * P)$.
           \ElsIf{$n(R*P))$ or can be merged in $G'$}\label{line:mergable}
             \State Replace $n(R)$ with $n(R * P)$. 
           \EndIf
         \EndFor
       \EndFor
       \State \Return $G'$.
     \EndFunction
   \end{algorithmic}
  \caption{Psuedocode for Preparation-on-Node Reduction Algorithm.} \label{algo:pnredux}
\end{figure}

\section{\PCAST Graph Optimizations} \label{sec:fullopt}
The previous sections outline the tools we leverage for \PCAST graph optimization. The full sequence of the \emph{\PCAST Graph Optimization Algorithm} is described in Fig.~\ref{algo:internop}. Note the conditional introduced at line~\ref{line:preprelease} is added specifically to target cases similar to Fig.~\ref{fig:motived}.

\begin{figure}
  \begin{algorithmic}[1]
    \Function{OptimizeGraph}{$G$, $F$, ``hold'' or ``release''}
      \State $G' = \emptyset,$ and $G''\gets G$.
      \While{$G' \neq G''$}
        \State $G' \gets$ \Call{ReduceNodesByPrep}{$G''$}
      \EndWhile
      \State Let $\Pi_{\text{e}}$ be the set of $\pprep$ nodes with outdegree $0$.
      \If{this is the ``hold'' case}
        \State $(\Pi_{\text{e}}, F) \gets$ \Call{ReduceFrameByPrep}{$\Pi_{\text{e}}, F$}.
      \ElsIf{this is the ``release'' case}
        \State Let $M_{\text{e}}$ be all $\mmeas$ nodes with outdegree $0$.
        \State Let $\C_{\text{e}}$ be the measurement space of $M_{\text{e}}$. 
        \State Remove all incomparable nodes of $M_{\text{e}}$ in $G'$.
        \State $F \gets F^0$ in $G'$.
        \If{$\Pi_{\text{e}} \neq \emptyset$} \label{line:preprelease}
          \State Let $\C_{\text{e}}$ be the measurement space of $M_{\text{e}}$.
          \State Remove $M_{\text{e}}$ from $G'$.
          \State $(F_{\text{s}}, C) \gets $\Call{FindStabilizers}{$M_{\text{e}}$}.
          \State Let $M_{\text{n}}$ be the single-qubit $\mmeas$ from $C$.
          \State Let $C_{\text{n}}$ be the measurement space of $M_{\text{n}}$.
          \State $(\Pi_{\text{e}}, F) \gets$\Call{ReduceFrameByPrep}{$\Pi_{\text{e}}, F_{\text{s}}$}.
          \State Add $M_{\text{n}}$ to $G'$.
          \State $(\underbar b, v) \gets$\Call{MapMeasurements}{$M_{\text{e}}, C, F_{\text{n}}$}
          \State Add the node $(v + \underbar b): \C_{\text{n}} \to \C_{\text{e}}$ to $G'$. 
        \EndIf
      \EndIf
      \State \Return $G'$
    \EndFunction
  \end{algorithmic}
  \caption{\PCAST Graph Optimization Algorithm}
\label{algo:internop}
\end{figure}

While Fig.~\ref{algo:internop} is used to optimize the \PCAST graph back to another \PCAST graph, there are other uses of the routines described herein. During circuit synthesis as described in \companion, the synthesis of non-Clifford nodes intentionally defers synthesis of measurement nodes topologically at the end of the graph in the ``release'' case and returns them. Note these nodes are clearly mutually commuting and so we are free to use Fig.~\ref{algo:stab} to synthesize the remaining circuit and Fig.~\ref{algo:mm} for the mapping of measurements. A benefit of this methodology is the stabilizer search is performed in the ``context'' of the rest of the circuit since the measurements are transformed by the residual Clifford. For example, consider the case of a VQE algorithm circuit where one first prepares an ansatz state and then measures part of the Hamiltonian using a mutually commuting set. The stabilizer search of the mutually commuting measurements in isolation is different than that of the search when combined with the (incomplete) circuit synthesis of the ansatz. As such, its not simply that we optimize the Clifford circuit at the interface of the ansatz preparation and the measurements, but the combined effect may change what Pauli operators are measured to build the desired outcome. We demonstrate this point in Fig.~\ref{fig:vqe_meas} and in the results of Section~\ref{sec:eval}.

\begin{figure} \label{fig:vqe_meas}
\centering

   \begin{tikzpicture}
        \draw[fill=blue!20, very thick, rounded corners, align=center] (0.5, -1.1) rectangle (1.55, 1.3);
        \node[align=center] at (1, 0.1)
            {Ansatz \\ Prep};

    \node[scale=0.6] at (3.5, 0) {
       \begin{tikzcd}[row sep=0.2cm]
           &\startingpoprgate{\Meas[c_0]{Z_0 Z_1}}
           &\startingpoprgate{\Meas[c_5]{Z_2 Z_3}}
           \\
           &\startingpoprgate{\Meas[c_1]{Z_0 Z_2}}
           &\startingpoprgate{\Meas[c_6]{X_0 X_1 X_2 X_3}}
           \\
           &\startingpoprgate{\Meas[c_2]{Z_0 Z_3}}
           &\startingpoprgate{\Meas[c_7]{Y_0 Y_1 Y_2 Y_3}}
           \\
           &\startingpoprgate{\Meas[c_3]{Z_1 Z_2}}
           &\startingpoprgate{\Meas[c_8]{X_0 X_1 Y_2 Y_3}}
           \\
           &\startingpoprgate{\Meas[c_4]{Z_1 Z_3}}
           &\startingpoprgate{\Meas[c_9]{Y_0 Y_1 X_2 X_3}}
           \\
       \end{tikzcd}
     };


        \draw[fill=blue!20, very thick, rounded corners, align=center] (-1.1, 2.9) rectangle (0, 5.15);
        \node[align=center] at (-0.55, 4)
            {Ansatz \\ Prep};    

    \node[scale=0.6] at (1.1, 4) {
       \begin{tikzcd}[row sep=0.2cm]
           &\startingpoprgate{\Meas[c'_0]{-Z_0 Z_1}}\arrow[r]&\\
           &\startingpoprgate{\Meas[c'_1]{Y_0 Y_1 X_2 X_3}}\arrow[r]&\\
           &\startingpoprgate{\Meas[c'_2]{-Z_2 Z_3}}\arrow[r]&\\
           &\startingpoprgate{\Meas[c'_3]{Z_0 Z_2}}\arrow[r]&\\
       \end{tikzcd}
       };

    \node[scale=0.6] at (4.15, 4) {
       \begin{tikzcd}[row sep=0.2cm]
         & \startingframegate{1}{
           \begin{aligned}
             \begin{bmatrix}
               c_0\\c_1\\c_2\\c_3\\c_4\\c_5\\c_6\\c_7\\c_8\\c_9
             \end{bmatrix}
             \mapsto
             \begin{bmatrix}
               1\\0\\1\\1\\0\\1\\0\\0\\0\\0
             \end{bmatrix}
             +
             \begin{bmatrix}
               1 & 0 & 0 & 0\\
               0 & 0 & 0 & 1\\
               0 & 0 & 1 & 1\\
               1 & 0 & 0 & 1\\
               1 & 0 & 1 & 1\\
               0 & 0 & 1 & 0\\
               1 & 1 & 0 & 0\\
               0 & 1 & 1 & 0\\
               1 & 1 & 1 & 0\\
               0 & 1 & 0 & 0
             \end{bmatrix}
             \begin{bmatrix}
               c'_0\\c'_1\\c'_2\\c'_3
             \end{bmatrix}
           \end{aligned} 
          }
       \end{tikzcd}
       };

     \draw [decorate,decoration={brace,amplitude=5pt, mirror, raise=4ex}]
  (1.75,-0.65) -- (5.65,-0.65);

    \draw [decorate,decoration={brace,amplitude=5pt, raise=4ex}]
  (.5, 0.85) -- (5.65, 0.85);

    \node at (3.25, 2) {\Large $\equiv^{\text{release}}$};
    \draw[->] (3, 2.2) -- (3, 2.55);

    \node at (4, -1.8) {\Large $\equiv^{\text{release}}$};
    \draw[->] (3.7, -2.05) -- (3.7, -2.45);
  
    \node[scale=0.6] at (0.75, -4) {
       \begin{tikzcd}[row sep=0.2cm]
           &\startingpoprgate{\Meas[c'_0]{-Z_0 Z_1}}\arrow[r]&\\
           &\startingpoprgate{\Meas[c'_1]{-Z_2 Z_3}}\arrow[r]&\\
           &\startingpoprgate{\Meas[c'_2]{X_0 X_1 X_2 X_3}}\arrow[r]&\\
           &\startingpoprgate{\Meas[c'_3]{-Z_1 Z_3}}\arrow[r]&\\
       \end{tikzcd}
       };

    \node[scale=0.6] at (3.9, -3.9) {
       \begin{tikzcd}[row sep=0.2cm]
         & \startingframegate{1}{
           \begin{aligned}
             \begin{bmatrix}
               c_0\\c_1\\c_2\\c_3\\c_4\\c_5\\c_6\\c_7\\c_8\\c_9
             \end{bmatrix}
             \mapsto
             \begin{bmatrix}
               1\\1\\0\\0\\1\\1\\0\\0\\0\\0
             \end{bmatrix}
             +
             \begin{bmatrix}
               1 & 0 & 0 & 0\\
               1 & 1 & 0 & 1\\
               1 & 0 & 0 & 1\\
               0 & 1 & 0 & 1\\
               0 & 0 & 0 & 1\\
               0 & 1 & 0 & 0\\
               0 & 0 & 1 & 0\\
               1 & 1 & 1 & 0\\
               0 & 1 & 1 & 0\\
               1 & 0 & 1 & 0
             \end{bmatrix}
             \begin{bmatrix}
               c'_0\\c'_1\\c'_2\\c'_3
             \end{bmatrix}
           \end{aligned} 
          }
       \end{tikzcd}
       };
  
     \end{tikzpicture}
  \caption{\PCAST graph representations for a hydrogen VQE example ($H_2$; JW encoding). The center graph represents the ansatz preparation and measurement of one mutually commuting set. The graph below represents the stabilizer search result for the measurements alone, whereas the top result represents the stabilizer search result in the context of the ansatz. Table~\ref{tab:comb_vs_sep} demonstrates the top outcome results in a better circuit.} \label{fig:vqe_meas}  
\end{figure}

\section{Evaluation}\label{sec:eval}
\begin{table*}[]
\centering
\caption{Statistics for the measurement circuits for each molecule for both BK and JW mappings.}
\label{tab:n_meas}
\resizebox{\textwidth}{!}{%
\begin{tabular}{@{}cccccclcclcclcclcc@{}}
\toprule
\multicolumn{1}{l}{} &
  \multicolumn{1}{l}{} &
  \multicolumn{1}{l}{} &
  \multicolumn{1}{l}{} &
  \multicolumn{2}{c}{Avg total gates} &
   &
  \multicolumn{2}{c}{Avg 2Q gates} &
   &
  \multicolumn{2}{c}{Avg Depth} &
   &
  \multicolumn{2}{c}{Avg $r_{2q}$} &
   &
  \multicolumn{2}{c}{Avg $k$} \\ \cmidrule(lr){5-6} \cmidrule(lr){8-9} \cmidrule(lr){11-12} \cmidrule(lr){14-15} \cmidrule(l){17-18} 
Benchmark & \# Qubits & \# Measurements & \# Groupings & BK    & JW    &  & BK    & JW    &  & BK    & JW    &  & BK    & JW    &  & BK    & JW    \\ \midrule
H$_2$        & 4         & 15              & 2            & 5.00  & 12.00 &  & 0.00  & 2.50  &  & 1.50  & 4.50  &  & -     & 2.40  &  & 4.00  & 3.50  \\
HF        & 12        & 631             & 38           & 41.61 & 38.34 &  & 9.79  & 8.89  &  & 10.53 & 7.89  &  & 6.53  & 7.12  &  & 10.26 & 9.82  \\
LiH       & 12        & 631             & 41           & 38.07 & 36.34 &  & 8.27  & 8.34  &  & 9.41  & 8.15  &  & 7.76  & 7.24  &  & 10.46 & 9.54  \\
BeH$_2$      & 14        & 666             & 36           & 47.33 & 43.83 &  & 11.47 & 10.36 &  & 11.17 & 9.69  &  & 7.53  & 8.40  &  & 11.31 & 11.31 \\
H$_2$O       & 14        & 1086            & 50           & 48.98 & 46.88 &  & 11.52 & 11.46 &  & 10.76 & 9.78  &  & 7.65  & 7.58  &  & 11.98 & 11.36 \\
BH$_3$      & 16        & 1957            & 76           & 48.85 & 53.15 &  & 10.29 & 13.11 &  & 10.45 & 10.03 &  & 11.47 & 8.93  &  & 14.40 & 13.57 \\
NH$_3$       & 16        & 2949            & 105          & 52.28 & 55.48 &  & 11.26 & 13.64 &  & 10.68 & 10.50 &  & 10.54 & 8.60  &  & 14.51 & 13.82 \\
CH$_4$       & 18        & 6892            & 163          & 71.14 & 59.24 &  & 17.58 & 13.58 &  & 13.89 & 9.54  &  & 8.60  & 11.04 &  & 16.74 & 16.06 \\
B$_2$        & 20        & 2239            & 64           & 84.96 & 72.70 &  & 23.94 & 20.23 &  & 16.00 & 13.06 &  & 6.31  & 7.42  &  & 16.64 & 16.26 \\
O$_2$        & 20        & 2239            & 67           & 85.89 & 72.84 &  & 24.57 & 20.16 &  & 15.55 & 12.95 &  & 7.46  & 9.04  &  & 16.43 & 16.18 \\
Be$_2$       & 20        & 2951            & 74           & 89.56 & 78.67 &  & 24.98 & 22.16 &  & 16.10 & 14.06 &  & 7.42  & 8.40  &  & 17.13 & 16.95 \\
C$_2$        & 20        & 2951            & 75           & 90.34 & 78.06 &  & 25.36 & 21.83 &  & 15.97 & 13.55 &  & 7.23  & 8.51  &  & 17.28 & 17.24 \\
F$_2$        & 20        & 2951            & 76           & 88.66 & 77.15 &  & 25.09 & 22.06 &  & 16.57 & 13.95 &  & 7.36  & 8.32  &  & 16.86 & 16.42 \\
Li$_2$       & 20        & 2951            & 78           & 90.81 & 77.22 &  & 25.87 & 21.96 &  & 17.04 & 13.67 &  & 7.12  & 8.40  &  & 16.87 & 16.58 \\
N$_2$        & 20        & 2951            & 79           & 93.53 & 79.00 &  & 27.10 & 22.10 &  & 16.90 & 13.91 &  & 6.86  & 8.41  &  & 17.35 & 17.06 \\ \bottomrule
\end{tabular}%
}
\end{table*}

We focus on the VQE use case and the use of Fig.~\ref{algo:stab} for measurement reduction. For comparison, we consider three methods found in the literature for solving a similar problem.\footnote{Unlike the references to follow, we only focus on the circuit generation, not the grouping problem for which they provide methods which are fully compatible with this work as demonstrated by our use of the grouping method described in Ref.~\cite{Crawford2021}.} Ref.~\cite{Verteletskyi2020} only considers the case of qubit-wise commutativity i.e. full qubit agreement as defined in Defn.~\ref{def:agree}. By virtue of the initial search for agreement in lines \ref{line:genstart}-\ref{line:genend}, Fig.~\ref{algo:stab} produces essentially the same circuit under these circumstances, thus making a direct comparison unnecessary. Ref.~\cite {Gokhale2020} allows for general commutativity, but their code implementation requires the commuting set have exactly $N$ independent elements (whereas our implementation allows one to specify any commuting set without considering independence). They also map the elements of the set, under conjugation by a Clifford, exactly to a single-qubit measurement, i.e. they do not leverage smaller weight bases for the stabilizer space.
Ref.~\cite{Crawford2021} is the least constrained and as such represents the best comparison to our methods; we use their method, \textsc{sorted insertion}, for partitioning the Hamiltonian, and compare \PCAST circuit constructions against their method.
Their methods achieve a ratio of the theoretical maximum number of two-qubit gates to actual number of two-qubit gates of approximately 3.5. Thus we compute a similar metric, 
\begin{align}\label{eq:r_2q}
  r_{\text{2q}} = \frac{Nk - k(k+1)/2}{\text{\# 2q gates}},
\end{align}
where $k$ is extracted as the number of single-qubit measurements in the PCOAST-optimized circuit.


\PCAST is implemented in C++ as the core optimization of the Intel$^{\text{\textregistered}}\xspace$ Quantum SDK, enabled by the (-O1) flag, which we use for all our experiments. We apply \PCAST to the Hamiltonians of several molecules in the context of VQE. We use Qiskit's PySCF driver~\cite{qiskit, pyscf} 
to express all fermionic Hamiltonians in the STO-3G basis as qubit observables using either the Bravyi-Kitaev (BK)~\cite{Bravyi_2002} or Jordan-Wigner (JW)~\cite{Jordan1928} transformation. Grouping of the Hamiltonian terms in mutually commuting sets is obtained by using an adapted version of the \textsc{sorted insertion} algorithm of Ref.~\cite{Crawford2021}. For all benchmarks in Table~\ref{tab:comb_vs_sep}, we pick the Unitary Coupled-Cluster Single and Double excitations (UCCSD)~\cite{uccsd} as our ansatz. Our experiments use an Intel Xeon$^{\text{\textregistered}}$ Platinum CPU ($2.4$GHz, $2$TB RAM). All averages are over the various groupings. 

\begin{table}[]
\centering
\caption{Average (\%) reductions in total gate count, two-qubit gates, depth, and measurement gates achieved when PCOAST synthesizes a combined ansatz and measurement circuit vs each individually. Results are shown for both BK and JW mappings.}
\label{tab:comb_vs_sep}
\resizebox{\columnwidth}{!}{%
\begin{tabular}{@{}cccccccccccc@{}}
\toprule
 &
  \multicolumn{2}{c}{\textbf{Total Gates}} &
   &
  \multicolumn{2}{c}{\textbf{2Q Gates}} &
   &
  \multicolumn{2}{c}{\textbf{Depth}} &
   &
  \multicolumn{2}{c}{\textbf{Meas}} \\ \cmidrule(lr){2-3} \cmidrule(lr){5-6} \cmidrule(lr){8-9} \cmidrule(l){11-12} 
\textbf{Benchmark} &
  \textbf{BK} &
  \textbf{JW} &
   &
  \textbf{BK} &
  \textbf{JW} &
   &
  \textbf{BK} &
  \textbf{JW} &
   &
  \textbf{BK} &
  \textbf{JW} \\ \midrule
H$_2$   & 70.00 & 35.27 &  & 0.00      & 41.34 &  & 50.00 & 5.45 &  & 75.00 & -16.67 \\
HF      & 5.89  & 11.54 &  & 3.92      & 10.64 &  & 3.89  & 6.56 &  & 7.11  & 3.62   \\
LiH     & 5.19  & 0.93  &  & 5.12      & -0.72 &  & 2.01  & 0.34 &  & 1.63  & 3.03   \\
BeH$_2$ & 1.90  & 3.99  &  & 0.62      & 2.69  &  & 1.75  & 0.64 &  & 2.91  & -7.25  \\
H$_2$O  & 2.81  & 5.19  &  & 2.52      & 4.33  &  & 4.26  & 2.18 &  & 3.55  & 0.62   \\
CH$_4$  & 4.09  & 3.16  &  & 3.84      & 2.97  &  & 1.04  & 0.32 &  & 1.04  & -1.70  \\
\bottomrule
\end{tabular}%
}
\end{table}

Table~\ref{tab:n_meas} shows our results. We find a combined average $r_{2q}$ value over all benchmarks of 7.85 (7.99) with a maximum value of 11.47 (11.04) for the BK (JW) encoding. All metrics appear to scale linearly in qubit number including $k$ with the ratio $\frac{k}{N}$ an average value of $83.64\%$. It is worth noting that in a majority of benchmarks, the JW encoding outperforms the BK encoding.  

We also consider the outcome of synthesizing the measurement circuit collectively with the ansatz using \PCAST. Table~\ref{tab:comb_vs_sep} shows the average reductions gained when PCOAST synthesizes a quantum program comprised of both ansatz and measurement circuits vs. when it synthesizes each part individually. Overall, we see that when  combined, PCOAST further reduces the total gates, two-qubit gates, depth, and measurement gates by $14.98\%$ ($10.01\%$), $2.67\%$ ($10.21\%$),  $10.49\%$ ($2.58\%$), and $15.21\%$ ($-3.06\%$) for the BK (JW) encoding.

Though there is a reduction in aggregate, we find several instance where the separate synthesis outperform the combined synthesis. Because the separate synthesis outcome is theoretically achievable with the combined synthesis, this demonstrates that there is more refinement available to both \PCAST in general and the Stabilizer Search Algorithm in particular. 

\section{Conclusion} \label{sec:conc}
In this paper, we have introduced a set of optimization routines included in \PCAST which optimize at the interface between unitary and non-unitary operations. This includes reduction of node cost and node number in the presence of preparation nodes via the Preparation-on-Node Reduction Algorithm, reduction of cost for Clifford operations in the presence of preparations via the Preparation-on-Frame Reduction Algorithm, and measurement cost reduction using Clifford operations via the Stabilizer Search Algorithm and the classical remapping of measurement outcomes via the Map Measurement Algorithms.
These routines are also combined to amplify their effectiveness when optimizing the \PCAST graph via the PCOAST Graph Optimization Algorithm. These routines were described in detail and shown to be valid relative to the appropriate equivalence (``hold'' or ``release''). We numerically studied the effectiveness of the Stabilizer Search Algorithm as used for measurement reduction circuits for VQE. We found an overall average ratio of the theoretical maximum number of two-qubit gates to resulting number of two-qubit gates of 7.85 (7.99) for the BK (JW) encoding. We also showed that synthesis of the VQE ansatz along with the measurements results in a better aggregate circuit.

In future work, we hope to further refine the algorithms presented here. In particular, we hope to increase the efficiency of the general Stabilizer Search Algorithm by considering \emph{partial} qubit masking. Instead of a fully opaque qubit mask as defined in Eq.\eqref{eq:mask}, we can weight the cost of the support based on the fraction of measurement elements which don't agree on that support. This should further reduce the number of TQE gates needed to find full qubit support agreement within the measurement set. We also look to expand the types of optimization routines within \PCAST, including more sophisticated transformations on unitary elements, and better use of stabilizer subspace methods by finding incomparable node cliques in the \PCAST graph.  

\bibliography{references}


\end{document}